\documentclass[10pt, conference, letterpaper]{IEEEtran}

\ifCLASSOPTIONcompsoc
  \usepackage[nocompress]{cite}
\else
  \usepackage{cite}
\fi
\usepackage{cite,url}
\usepackage{booktabs}
\usepackage{amsmath,amssymb,bm,dsfont,bbm,mathtools}
\usepackage{algorithmic}
\usepackage{amsthm,nccmath,lipsum,cuted}
\usepackage[dvips]{color}
\usepackage{graphicx}
\usepackage{xcolor}
\usepackage[lined,boxed,commentsnumbered, ruled]{algorithm2e}
\usepackage{subcaption}
\usepackage{cleveref}
\usepackage{amsfonts}
\usepackage{ragged2e}
\newtheorem{theorem}{Theorem}

\newtheorem{remark}{Remark}
\newtheorem{corollary}{Corollary}

\usepackage{graphicx}
\usepackage{verbatim}
\usepackage[footnote,draft,danish,silent,nomargin]{fixme}

\def\P{{\mathbb P}}  

\def\Y{{\mathcal Y}}

\graphicspath{{./figs/}}

\IEEEoverridecommandlockouts

\begin{document}

\title{Performance Characterization Using AoI in a Single-loop Networked Control System}

\author{\IEEEauthorblockN{Jaya Prakash Champati\IEEEauthorrefmark{1}, 			     Mohammad H. Mamduhi\IEEEauthorrefmark{2}, Karl H. Johansson\IEEEauthorrefmark{2},
    James Gross\IEEEauthorrefmark{1}}
    \IEEEauthorblockA{\IEEEauthorrefmark{1} Division of Information Science and Engineering, EECS, KTH Royal Institute of Technology, Stockholm, Sweden}
    \IEEEauthorblockA{\IEEEauthorrefmark{2} Division of Automatic Control, EECS, KTH Royal Institute of Technology, Stockholm, Sweden}
$\{$jpra,mamduhi,kallej,jamesgr$\}$@kth.se
\thanks{This work has been partially supported by the Swedish Research Council VR under grant 2016-04404, and Knut and Alice Wallenberg Foundation.}
}

\maketitle

\begin{abstract}
The joint design of control and communication scheduling in a Networked Control System (NCS) is known to be a hard problem. Several research works have successfully designed optimal sampling and/or control strategies under simplified communication models, where transmission delays/times are negligible or fixed. However, considering sophisticated communication models, with random transmission times, result in highly coupled and difficult-to-solve optimal design problems due to the parameter inter-dependencies between estimation/control and communication layers. To tackle this problem, in this work, we investigate the applicability of Age-of-Information (AoI) for solving control/estimation problems in an NCS under i.i.d. transmission times. Our motivation for this investigation 
stems from the following facts: 1) recent results indicate that AoI can be tackled under relatively sophisticated communication models, and 2) a lower AoI in an NCS may result in a lower estimation/control cost. 
We study a joint optimization of sampling and scheduling for a single-loop stochastic LTI networked system with the objective of minimizing the time-average squared norm of the estimation error. We first show that  under mild assumptions on information structure the optimal control policy can be designed independently from the sampling and scheduling policies. We then derive a key result that minimizing the estimation error is equivalent to minimizing a function of AoI when the sampling decisions are independent of the state of the LTI system. 
Noting that minimizing the function of AoI is a stochastic combinatorial optimization problem and is hard to solve, we resort to heuristic algorithms obtained by extending existing algorithms in the AoI literature. We also identify a class of LTI system dynamics for which minimizing the estimation error is equivalent to minimizing the expected AoI.

\end{abstract}

\section{Introduction}\label{sec:intro}

In the recent past, there has been an ever increasing interest in studying Networked Control Systems (NCS) that support time-critical-control-loop applications which include, among many others, smart grids, Internet-of-Things (IoT), sensor networks and augmented reality \cite{Baillieul_IEEE,Gupta_IEEE}. In such applications a status update that is received after certain duration of its generation time may become stale at the receiver and the control decision taken based on this stale sample may lead to untimely feedback and hence undesired control action. Thus, the freshness of the status updates at the receiver plays a key role in the design of such systems wherein time is more critical at the receiving end. 

Even though NCSs have been studied extensively from the control perspective focusing on optimizing control performance, majority of the works have considered designing sampling and/or control strategies over simplified networking/communication models with idealized assumptions wherein status updates are assumed to have zero or constant transmission delay. Even for the fundamental problem of Minimum Mean Square Error (MMSE) estimation, results are scarce for computing optimal event-based sampling strategies when the transmission delays in the network between the sampler and the estimator are i.i.d. This can be attributed to the complicated parameter inter-dependencies between estimation/control and communication layers, that arise due to the end-to-end delay, resulting in a highly coupled and difficult-to-solve optimal design problems~\cite{Chiang_IEEE}. 
However, as the emerging networked control applications are envisioned to be running on the edge in future wireless networks, it is necessary to consider more realistic communication models wherein the transmission delays in the network are non-negligible and random. 

Recently, the Age of Information (AoI) metric, proposed in~\cite{kaul_2011a}, has emerged as a novel metric to quantify the freshness of the received status updates and has attracted significant attention from communication and networking community. It is defined as the time elapsed since the generation of the latest successfully received status update at the destination. 
Several works have studied the problem of minimizing some function of AoI under different queuing and communication models~\cite{kaul_2012b,yates_2012a,yates_2015a,Champati_2018a,Sun_2017}. While the works in~\cite{kaul_2012b,yates_2012a,yates_2015a} consider time averaged AoI, the authors in~\cite{Champati_2018a} consider minimizing the tail of the AoI, and the authors in~\cite{Sun_2017} consider any non-decreasing and measurable function of AoI. Apart from studying the effects of communication scheduling on AoI, none of the above works consider estimation/control objectives in networked systems. Nonetheless, we would like to note that:
\begin{enumerate}
\item All the above works minimize some function of AoI assuming the transmission times are either i.i.d. or Markovian.
\item A general consensus is that, a lower AoI in an NCS may result in a lower estimation/control cost, because having access to fresher state information often improves the estimation/control performance.
\end{enumerate}
While item 1) suggests that AoI could be tackled under relatively sophisticated communication models, item 2) suggests that the solutions proposed for AoI could be considered for studying estimation/control costs. 
Given the above facts, the question we would like to pursue is whether the scheduling strategies proposed in the AoI literature, under sophisticated communication models, could be used or extended to minimize estimation/control costs in an NCS. Answering this question will not only shed light on the potential use of scheduling strategies proposed in the AoI literature for networked control, but further motivates the work on the AoI metric under more realistic communication models. However, to achieve this, we need a precise understanding of the relationship between AoI and estimation/control costs in networked systems. 

The authors in~\cite{Sun2017_a} have studied the MMSE problem with i.i.d. transmission delays for Wiener process estimation. They have shown that the estimation error is a function of AoI if the sampling decisions are independent of the observed Weiner process; otherwise, the estimation error is not a function of AoI.
Following this line of research, in this paper we examine the relation between AoI and a typical control cost for a Linear-Time-Invariant (LTI) system. In particular, we consider a single-loop LTI stochastic networked system, shown in Figure~\ref{fig:NCS}, where the next state of the system is a linear function of the current state, the control input and an associated Gaussian noise. The system is equipped with an event-based sensor that decides about the next sampling instant. The samples/status updates are delivered to an estimator by a communication link having random transmission time per status update. The freshness of the status updates received at the estimator is a function of the sampling policy, communication scheduling policy, and the distribution of the transmission times. Using the received updates, the estimator computes the current state of the plant and feeds it to a controller that computes the control input, which is instantaneously available at the actuator and the feedback loop is closed. 

\begin{figure}[t]
\centering
\includegraphics[width = 3.2in]{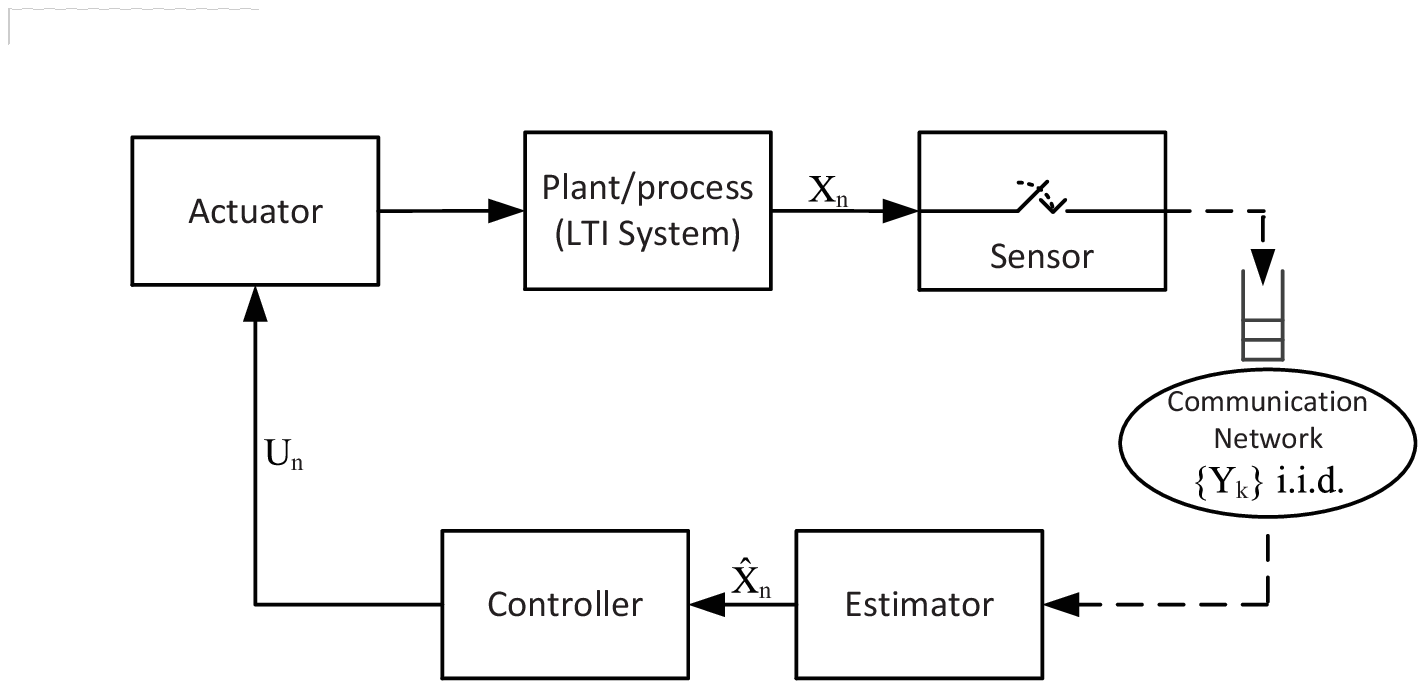}
\caption{Single-loop network control system.}
\label{fig:NCS}
\end{figure}

We first note that under some mild assumptions on the information structures of the control unit and the sampling unit, the optimal control policy can be designed independently from the sampling policy, and then the optimal sampling policy minimizes a function of mean square estimation error \cite{Maity_LCSS}. We then show that the latter objective function can be expressed as a function of AoI. Having these, we derive the optimal design for the single-loop NCS with the cost as a function of AoI. Noting that minimizing the function of AoI in our problem setting is a stochastic combinatorial optimization problem and is hard to solve, we study heuristic algorithms by extending existing algorithms in the AoI literature~\cite{Sun_2017}.

The rest of this paper is organized as follows. Section II introduces the problem statement and the overall system model. Section III presents the main results on performance characterization in terms of AoI and heuristics to solve the problem. Numerical results are presented in Section IV, and we conclude the paper in Section V. 

\section{Problem Statement}
We study a single-loop NCS, shown in Figure~\ref{fig:NCS}, where a single non-scalar controlled linear time-invariant system closes the loop between the plant (collocated with the sensors) and the controller (collocated with the estimator) through a communication network. The link from the controller to the plant is supposed to be a direct error-free connection sending the control signals to the actuator in a timely manner. The system time is slotted and $n \in \mathbb{Z}_{\geq 0}$ denotes a time slot, where $\mathbb{Z}_{\geq 0}$ is the set of non-negative integers.  The LTI system evolves linearly as follows:
\begin{align}\label{eq:state}
X_{n} = A X_{n-1} + B U_{n-1}+W_{n-1},
\end{align}
where $X_n\in\mathbb{R}^d$ is the state of the system at time slot $n$, $d$ is the system dimension, $U_n\in\mathbb{R}^q$ represents the control input, and $W_n \sim \mathcal{N}(0,\Sigma)\in\mathbb{R}^d$ is the exogenous noise having multi-variate Gaussian distribution with zero mean and covariance $\Sigma$. 
The constant matrices  $A\in \mathbb{R}^{d\times d}$ and $B \in \mathbb{R}^{d\times q}$ represent the system and input matrices, respectively. The noise realizations are assumed to be mutually independent and identically Gaussian distributed, hence, we can re-write $\Sigma = \sigma^2 I_d$, where $I_d$ is a $d\times d$ identity matrix and $\sigma^2$ is the variance. In addition, we assume that the noise realizations are independent from the initial condition $X_0$, which itself is presumed to be selected from any arbitrary random distribution with finite moments.  

A sensor samples the state of the system according to a sampling policy $\mathbf{g}$, which specifies the time slots in which samples are generated. The samples are submitted to a work-conserving server, e.g., a communication link, which can store them in a queue. The server transmits the samples/packets to the control unit, that consists of an estimator and a controller, using a non-preemptive scheduling policy $\mathbf{\pi}$. The scheduling policy $\mathbf{\pi}$ decides when a packet that is queued will be transmitted. The transmission link between the controller and the actuator is assumed to be perfect, i.e., the actuator receives the control signal as soon as the control input is executed. We use $k$ to denote the index of a sample that is received $k$th in the order at the receiver. Let $Y_k \in \{0,1,2,\ldots\}$ denote the transmission time of the sample $k$. We assume that $Y_k$'s are i.i.d. for all $k$, and $\mathbb{E}[Y] < \infty$. We use $D_k$ to denote the departure time of sample $k$.

In this work, we consider both sampling and scheduling policies are \textit{stationary randomized} policies. 
Under a stationary randomized policy, the decision in slot $n$ is completely determined by the transmission time of the most recently transmitted packet and some fixed probability measure. We formalize this definition shortly in Section~\ref{sec:equivalence}. Let $\mathcal{G}_\text{SR}$ and $\Pi_\text{SR}$ denote the set of all stationary randomized sampling policies and stationary randomized scheduling policies, respectively. 

In time step $n$, let $\delta_n$ represent the sampler's decision either to sample the state of the system or not, i.e., 
\begin{equation*}
   \delta_n=\begin{cases} 1, & \qquad X_n \;\text{is sampled}, \\ 0, & \qquad  \text{otherwise}. \end{cases}
  \end{equation*}

At the control unit, we denote the age of information at time $n$ by
$\Delta_n = n - T_n$, where $T_n$ is the most recent generation time of any packet that is received by time $n$. Note that $\Delta_n$ depends on the sampling policy. Also, we note that $\Delta_n$ increases by a unit step in each time slot until the departure of some packet and it drops to a value equal to the system delay of that packet. 

In \cite{6475979, Moli_CE} the necessary conditions for the optimal controller to be of the certainty equivalence (CE) form are derived for the NCS scenario with state-based sensor data sampling. Inspired by the mentioned works, we use their results to show that the optimal control policy in our problem can be considered to be in the class of CE control under some mild assumptions on the information structure. It is discussed in \cite{6475979} that the optimal control is CE if 1) the sampling decision $\delta_n$ at every time-step $n$ is independent of the applied control inputs $\{U_0,\ldots,U_{n-1}\}$, and, 2) an error-free instantaneous acknowledgement channel exists between control and sampling sides to inform the sampler about data delivery status. More precisely, the mentioned conditions guarantee independence of the sampling decisions from the control actions. Assuming that the sampler has access to the full plant information together with the instantaneous data delivery acknowledgment from the controller, the sampler can reconstruct the statistics of the input signals by an estimator installed at the sampling side. As we confine our attention to the class of stationary randomized policies, the sampler can be assumed to be dependent on the primitive parameters of the system, i.e., $\{X_0, W_0,\ldots, W_n\}$, the constant parameters $A, B, \Sigma$, and the scheduling outcome through the acknowledgement while any dependency on the control inputs appears in statistical form computed at the sampler. Therefore, we can rely on the results of \cite{6475979, Moli_CE} on the optimality of CE controller assuming that the sampler has access to the mentioned information set and a timely acknowledgement signal is  exchanged between controller and sampler. 



The optimal control policy can then be expressed in form of certainly equivalence control as follows: 
\begin{equation}\label{control-CE}
U_n= -K_n \mathbb{E}[X_n|\mathcal{I}_n],
\end{equation}
where $K_n$ is the optimal control gain, and $\mathcal{I}_n=\{X_0,U_0,\ldots,U_{n-1}, \Delta_n, X_{n-\Delta_n}\}$ is the controller information history at time-step $n$. It should be noted that any possible design of the sampler neither affects the control law nor the control gain $K_n$, but affects the estimation precision $\mathbb{E}[X_n|\mathcal{I}_n]$. Moreover, due to the independence of the design of the control law from the sampling process, we can design the control input optimally according to the desired optimality criteria, e.g. LQG control (see \cite{Moli_CE} for detailed derivation of the optimal control gain based on LQG control). This affects only the control gain $K_n$, and we omit it for the purpose of brevity. The computed control signal is fed back to the actuator and we can then express the closed-loop dynamics as follows:
\begin{equation*}
X_n\!=\! \left(A-BK_{n-1}\right)X_{n-1}+BK_{n-1}\!\left(X_{n-1} - \hat{X}_{n-1}\right)+W_{n-1}.
\end{equation*}

The estimator at the control side computes the estimation $\mathbb{E}[X_n|\mathcal{I}_n]$. To derive the dynamics of the estimator, let us first derive the dynamics of the system state $X_n$ as a function of the defined age of information $\Delta_n$. Using~\eqref{eq:state} we conclude
\begin{align}\label{eq:state1}
X_{n} &= A^{\Delta_n}X_{n - \Delta_n} + \sum_{j=1}^{\Delta_n}A^{j-1}W_{n-j}\\\nonumber
&+B U_{n-1}+AB U_{n-2}+\ldots+A^{\Delta_{n-1}}B U_{n-\Delta_n}
\end{align}

Taking the expectation from the expression (\ref{eq:state1}) conditioned on $\mathcal{I}_n$, the estimated system state $\hat{X}_n$ can be expressed as
\begin{align}\label{eq:estimator}
\hat{X}_n &= \mathbb{E}[X_n|\mathcal{I}_n] = A^{\Delta_n}X_{n - \Delta_n}\\\nonumber
&+B U_{n-1}+AB U_{n-2}+\ldots+A^{\Delta_{n-1}}B U_{n-\Delta_n}.
\end{align}
Using~\eqref{eq:state1} and~\eqref{eq:estimator}, we can compute the estimation error $e_n$ at the estimator side as follows:
\begin{align}\label{eq:error}
e_n = X_n - \hat{X}_n = \sum_{i=1}^{\Delta_n}A^{i-1}W_{n-i}.
\end{align}

We are interested in finding stationary randomized policies $\mathbf{g}$ and $\mathbf{\pi}$ that minimize the time-average squared error norm, i.e., we aim to solve $\mathcal{P}$, where
\begin{align*}
\mathcal{P}:\; \underset{(\mathbf{g}\in \mathcal{G}_\text{SR},\mathbf{\pi}\in \Pi_\text{SR})}{\text{minimize}} \lim_{n \rightarrow \infty} \frac{1}{n} \sum_{j=0}^{n-1}\|e_j\|_2^2.
\end{align*}
Our goal is to establish a concrete relation between the objective function of $\mathcal{P}$\footnote{In the domain of stationary randomized policies considered in this paper, $\Delta_n$ is stationary and ergodic. This is also true for the error process $e_n$, given in~\eqref{eq:error}, which turns out to be stationary and ergodic. Thus, the limit in the objective of $\mathcal{P}$ exists.}  and $\Delta_n$. In particular, we show that solving $\mathcal{P}$ is equivalent to minimizing a specific function of $\Delta_n$. Given this equivalence, we use algorithms from the AoI literature to solve $\mathcal{P}$.

\section{Joint Optimization of Sampling and Scheduling Policies}\label{sec:equivalence}
Note that under the sampling policy $\mathbf{g}$, a sample can be generated in any time slot. However, there is no advantage in generating samples and storing them in a queue while a sample is being transmitted. To see this, from~\eqref{eq:estimator} and~\eqref{eq:error} we infer that the estimation of the system state using a recent state results in a lower error than that of using an older state. Thus, when the transmission of a packet is finished, sampling the current state and transmitting the packet results in lower error than transmitting older packets stored in the queue. Therefore, $\mathbf{\pi}$ is degenerate and the samples are never queued. The  sampling policy then can be defined as $\mathbf{g} \triangleq \{G_k, k \geq 1\}$, where $G_k$ is a decision variable which represents the number of time slots the system waits before generating a new sample $k$ after the transmission of sample $(k-1)$. We assume that $G_k \in \{0,1,\ldots,M\}$ for all $k$, where $M < \infty$ denotes the maximum waiting time tolerated by the system\footnote{For the sake of simplicity in exposition we abuse the notation by using $G_k$, which actually is a mapping from the domain of the sampling policy $\mathbf{g}$, at the departure instant of sample $k$, to $\{0,1,\ldots,M\}$.}. 

Under a causal sampling policy $\mathbf{g}$, $G_k$ is determined by the observations $\{Y_i, i\leq k\}$ and previous decisions $\{G_i, i \leq k - 1\}$. Let $\mathcal{G}$ denote the set of all causal policies. A stationary randomized policy is a causal policy under which $G_k$ is assigned a value from $\{0,1,\ldots,M\}$ based on $Y_{k-1}$ and a fixed probability measure. In the following we formally define the stationary randomized policies $\mathcal{G}_\text{SR}$.
\begin{align*}
\mathcal{G}_\text{SR} \triangleq \{& \mathbf{g} \in \mathcal{G}: \forall a \in \{0,1,\ldots,M\}, \P\{G_k \leq a|Y_{k-1} = y\} \text{ is} \\ &\text{independent of} \;k\}.
\end{align*}


\begin{theorem}\label{thm:equivalence:event}
$\mathcal{P}$ is equivalent to $\tilde{\mathcal{P}}$, almost surely, where
\begin{equation*}
\begin{aligned}
& \tilde{\mathcal{P}}:&& \min_{\mathbf{g} \in \mathcal{G}_\text{SR}} \quad \mathbb{E}[f(\Delta)]\\
& \text{s.t.} && G_k \in \{0,1,\ldots,M\},\, \forall k. 
\end{aligned}
\end{equation*}
where the function $f:\mathbb{Z}_{\geq 0} \rightarrow \mathbb{R}^{+}$, and is given by,
\begin{align*}
f(\Delta) = \sum_{i=0}^{\Delta-1} \textsf{Tr}\left(A^{i^\top}A^{i}\Sigma\right),
\end{align*}
$\textsf{Tr}(\cdot)$ is the trace operator, and
\begin{align}\label{eq:ExpfDelta}
\mathbb{E}[f(\Delta)] = \frac{\mathbb{E}\left[\sum_{j=Y_{k-1}}^{Y_{k-1} + G_k + Y_k - 1} \sum_{i=0}^{j-1} \textsf{Tr}\left(A^{i^\top}A^{i}\Sigma\right)\right]}{\mathbb{E}[Y_k + G_k]}.
\end{align}
\end{theorem} 
\begin{proof}
The proof is given in the Appendix.
\end{proof}
The result in Theorem~\ref{thm:equivalence:event} asserts that minimizing the time averaged square error norm is equivalent to minimizing the expected value of a specific function of AoI, i.e. $f(\Delta)$, with parameters $A$ $\Sigma$. In the following corollary we present a condition under which the solution of $\tilde{\mathcal{P}}$ is equivalent to minimizing the expected AoI $\mathbb{E}[\Delta]$.

\begin{corollary} \label{cor:equivalence}
The optimal solution of $\tilde{\mathcal{P}}$ is equivalent to minimizing the expected AoI $\mathbb{E}[\Delta]$ if there exists a constant $\gamma\in \mathbb{R}^+$ such that
\begin{align}\label{eq:equivalence:cond}
f(j) = \sum_{i=0}^{j-1} \textsf{Tr}\left(A^{i^\top}A^{i}\Sigma\right) = \gamma j,\,\, \forall j.
\end{align}
\end{corollary}
\begin{proof}
One can show that the expected AoI for sampling policy $\mathbf{g} = \{G_k,k\geq 1\}$ is given by,
\begin{align}\label{eq:expAoI}
\mathbb{E}[\Delta] = \frac{\mathbb{E}\left[\sum_{j=Y_{k-1}}^{Y_{k-1} + G_k + Y_k - 1} j \right]}{\mathbb{E}\left[G_k + Y_k \right]}.
\end{align}
Therefore, the result follows by substituting~\eqref{eq:equivalence:cond} in~\eqref{eq:ExpfDelta}, and~\eqref{eq:expAoI}.
\end{proof}

\begin{remark}
The class of orthogonal matrices satisfy condition (\ref{eq:equivalence:cond}) and therefore, meet the requirement of Corollary \ref{cor:equivalence}. To show this, we recall that any orthogonal matrix $M$, $M^\top=M^{-1}$ holds, and hence $M^\top M=I$. In addition, $(M^r)^\top=(M^\top)^r$ holds for orthogonal matrices. Having these equalities together with the condition (\ref{eq:equivalence:cond}) results in
\begin{align*}
\sum_{i=1}^j \textsf{Tr}\left(M^{i^\top}M^{i}\Sigma\right)&=\sum_{i=1}^j \textsf{Tr}\left((\underbrace{M^{\top}M}_{I})^{i}\Sigma\right)\\
&=\sum_{i=1}^j \textsf{Tr}\left(\Sigma\right)=j\;\textsf{Tr}\left(\Sigma\right)
\end{align*}
Hence, there always exists $\gamma=\textsf{Tr}\left(\Sigma\right)$, $\forall j$. It is worth noting that for the scalar case, $A \in \{1,-1\}$ are the only values that satisfy condition (\ref{eq:equivalence:cond}).
\end{remark}

\begin{remark}
Assuming that the $A$ matrix is orthogonal results in having a stochastic lossless open-loop linear system, i.e. removing the control from the equation (\ref{eq:state}), any increase in the mean size of the system state would correspond to stochastic noise. In fact, without any control input $\bar{U}_{n-1}$, we have $\mathbb{E}[\|X_{n}\|^2]=\mathbb{E}[\|X_{n-1}\|^2]+\Sigma$. This means even without control, such a stochastic system is mean-square bounded if $\Sigma$ and $X_0$ are both bounded. Having control results naturally in tighter bounds. Removing noise from the system (\ref{eq:state}), i.e having a deterministic system, and if $A$ is orthogonal we have a linear system that is marginally stable, i.e. $\|X_{n}\|^2=\|X_{n-1}\|^2$. Adding control to the noiseless system with orthogonal $A$ matrix, under controllability of the pair $(A,B)$, any stabilizing controller results in an asymptotically stable closed-loop system. 
\end{remark}

\subsubsection*{Heuristic Algorithms for Solving $\tilde{\mathcal{P}}$}
From Theorem~\ref{thm:equivalence:event} solving $\tilde{\mathcal{P}}$ is equivalent to solving $\mathcal{P}$, and thus we aim to solve $\tilde{\mathcal{P}}$. However, since $G_k$ are discrete variables, $\tilde{\mathcal{P}}$ is a stochastic combinatorial optimization problem, which is hard to solve, in general. Therefore, we aim for good heuristic solutions by first solving $\tilde{\mathcal{P}}$ by replacing the constraint with $G_k \in [0,M]$, and then round the resulting $G_k$ values. To this end, we use the algorithms proposed in~\cite{Sun_2017}\footnote{Recently, the authors in~\cite{Sun_2018a} proposed improved algorithms for a more generalized version of the problem studied in~\cite{Sun_2017}.}, where the authors have studied the problem of minimizing the expected value of any non-negative and non-decreasing function of AoI with $G_k \in [0,M]$. Note that $f(\cdot)$ is non-negative and non-decreasing function, however, the optimality results in~\cite{Sun_2017} are not directly applicable to $\tilde{\mathcal{P}}$ as the decision variable $G_k$ are discrete. The heuristic solutions we study are described below. 
\begin{itemize}
\item \textbf{Minimizing Expected AoI Solution (MEAS)}: In this heuristic we use [Algorithm 2,~\cite{Sun_2017}] to compute $G_k$. In [Algorithm 2,~\cite{Sun_2017}], whenever the transmission of packet $(k-1)$ is completed, $G_k$ is computed using the observed transmission time $y$ of that packet, or more formally using $Y_{k-1} = y$. For the sake of completeness, in Algorithm 1, we present the details of computation of $G_k$. Algorithm 1 outputs continuous values for $G_k$, i.e., $G_k \in [0,M]$ for all $k$. We round the values by using the floor function.
\begin{algorithm}
\begin{algorithmic}[1]
\caption{Algorithm for computing $G_k$}
\label{algo1}
\STATE \textbf{given} $l=0$, sufficiently large $u$, tolerance $\epsilon$
\REPEAT
\STATE $\beta \coloneqq (l+u)/2$
\STATE $G_k = \max(\beta-y,0)$
\STATE $o \coloneqq \mathbb{E}\left[(Y_k+G_k)^2 \right] - 2\beta \mathbb{E}[Y_k + G_k]$. 
\STATE \textbf{if} $o \geq 0$, $u \coloneqq \beta$; \textbf{else}, $l \coloneqq \beta$.
\UNTIL{$u - l \leq \epsilon$}
\STATE Compute $G_k = \max(\beta-y,0)$.
\STATE \textbf{return} $G_k$ 
\end{algorithmic}
\end{algorithm}

\item \textbf{Zero-wait policy}: Under the zero-wait policy, a new packet is generated and immediately transmitted only when the previous packet finishes its transmission.
\end{itemize}
We note that [Algorithm 1, \cite{Sun_2017}] can also be used to obtain a heuristic solution for $\mathcal{P}$. However, in our simulation we found that any trivial implementation of an extension of [Algorithm 1, \cite{Sun_2017}] for finding a heuristic solution for $\tilde{\mathcal{P}}$ results in a solution equivalent to zero-wait policy. Therefore, we leave the non-trivial extension of that algorithm for future work.

\section{Numerical Results}
In this section, we present our initial numerical results, where we study MEAS and zero-wait policies, and compare their performance under geometric transmission time distribution with success probability $p$. We consider single dimension system $d=1$, i.e., $A$ is a scalar. We observed similar results for $d=2$ which are not presented due to redundancy. 

In Figure~\ref{fig:solution_EventBased}, we present the solution provided by MEAS by plotting $G_k$ versus $y$, for different $p$ values. We truncate the x-axis values at $y = 10$. We observe that for smaller values of $p$, waiting times are larger. For example, when $p = 0.1$, we compute $G_k > 0$, for all $y \leq 8$, and $G_k = 7$, if $Y_k = 1$. To interpret this, when the probability of success is low and if the previous transmission happens successfully in fewer time slots than expected, then it is beneficial to wait before the next transmission. 
\begin{figure}[t]
\centering
\includegraphics[width = 3.2in]{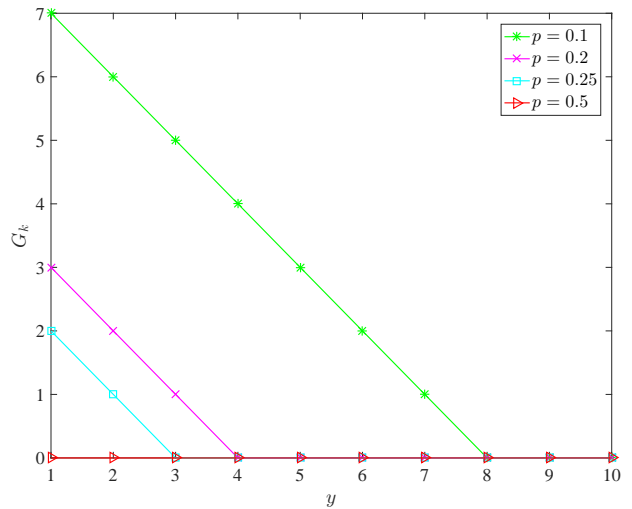}
\caption{Waiting time function $g(y)$ of MEAS algorithm.}
\label{fig:solution_EventBased}
\end{figure}

Under zero-wait policy, the expected AoI for geometric service times can be computed as follows:
\begin{align*}
\mathbb{E}[\Delta]& = \frac{\mathbb{E}[Y_k^2]}{2\mathbb{E}[Y_k]} + \mathbb{E}[Y_k] - \frac{1}{2}\\
&= \frac{4-p}{2p} - \frac{1}{2}.
\end{align*}
We note that the constant $\frac{1}{2}$ appears in the above expression as $\Delta_n$ in our system model is a discrete-time function. 
In Table~\ref{tabel1}, we compare the time-average squared error achieved by MEAS and zero-wait policy for $A = 1$. Since $A = 1$  satisfies condition~\eqref{eq:equivalence:cond}, for this case solving $\mathcal{P}$ is equivalent to minimizing $\mathbb{E}[\Delta]$ (cf. Corollary~\ref{cor:equivalence}), and we expect the solution provided by MEAS to be near optimal. Note that MEAS provides lower value than zero-wait policy, but the difference is not significant.
\begin{table}[ht]
\renewcommand{\arraystretch}{1.2}
\caption{Comparison of time-average squared error under zero-wait policy and MEAS algorithm for $A = 1$.}
\centering
\begin{tabular}{l|c|c}

p & Zero wait & MEAS\\
\hline
0.01  & 199 & 189.15 \\
\hline
0.05 & 39 & 37.22 \\
\hline
0.1 & 19 & 18.21 \\
\hline
0.2  & 9 & 8.70 \\
\hline
0.4 & 4 & 4 \\
\hline
0.8 & 1.5 & 1.5 \\
\end{tabular}
\label{tabel1}
\end{table}

In Figure~\ref{fig:comparison_alpha_lessthan_1}, we compare MEAS and zero-wait policy for $A$ values less than $1$. We again observe the same trend as before. We conclude that, for geometric transmission time distribution zero-wait policy is favorable as it has lower computational complexity and achieves time-average squared error that is negligibly close to that of MEAS. Finally, we note that for $A$ values greater than $1$, the time-average squared error approaches infinity.
\begin{figure}[h]
\centering
\includegraphics[width = 3.2in]{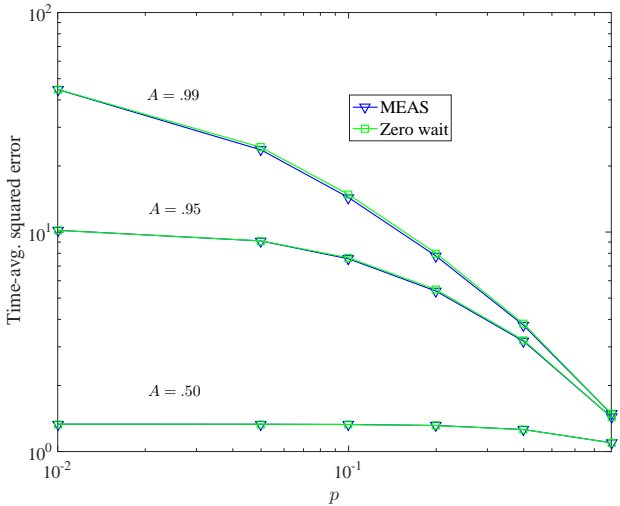}
\caption{Time average squared error under different algorithms for varying $p$ and different $A$ values.}
\label{fig:comparison_alpha_lessthan_1}
\end{figure}

\section{Conclusions and Future Work}
The motivation for this work is to investigate the applicability of the strategies/solutions proposed in the AoI literature to typical estimation/control problems in NCSs. As a first step, in this work, we have studied the joint optimal sampling and scheduling policies of a single-loop stochastic LTI networked
system. The main contribution has been in characterizing the estimation error as a function of AoI.
More precisely, for the given system model with a single source, i.i.d. service/transmission times and event-based sampling, we have shown that minimizing the time-average squared norm of the estimation error is equivalent to minimizing $\mathbb{E}[f(\Delta)]$, where $f$ is a non-negative and non-decreasing function with parameters $A$ (the LTI matrix) and $\Sigma$ (the noise covariance). 
Moreover, it is discussed that certainty equivalence controllers can be realized under certain assumptions
on the information structures available for the control unit and
the sampling unit. Further, we have provided a condition for $A$ under which minimizing the time averaged square norm is equivalent to minimizing expected AoI. Noting that minimizing $\mathbb{E}[f(\Delta)]$ is a stochastic combinatorial optimization problem, we resort to studying heuristic algorithms MEAS (an extension of Algorithm 2~\cite{Sun_2017}) and the zero-wait policy. Our initial numerical results show that, under geometric service time distribution, the zero-wait policy, despite having lower computational complexity, achieves only slightly higher estimation error in comparison with MEAS.

We leave the comprehensive numerical analysis studying different heuristic solutions under different service-time distributions for future work. We also aim to find heuristics solutions with performance guarantees, for solving $\tilde{\mathcal{P}}$. 


\appendix
\subsection{Proof of Theorem~\ref{thm:equivalence:event}}
Since $\mathbb{E}[Y] < \infty$, and $Z_k < M$ for all $k$, number of samples transmitted goes to infinity as $n$ goes to infinity. Therefore, we divide the time line using inter-departure times $\{D_k-D_{k-1},k\geq 1\}$, and reformulate the objective function as follows: 
\begin{align}\label{eq1:thm1}
\lim_{n \rightarrow \infty} \frac{1}{n} \sum_{j=0}^{n-1}\|e_j\|_2^2 &= \lim_{K \rightarrow \infty} \frac{1}{D_K} \sum_{k=1}^{K}\sum_{j=D_{k-1}}^{D_k-1}\|e_j\|_2^2,  
\end{align}
where $D_0 = 0$, and the inter-departure time is given by
\begin{align*}
D_k - D_{k-1} = Y_k + G_k.
\end{align*}
Recall that $Y_k$'s are i.i.d. Under a stationary randomized policy, $G_k$ only depends on $Y_{k-1}$, and thus $G_k$'s are also i.i.d. Therefore, $(D_k - D_{k-1})$'s are i.i.d as well, $\forall k$. 

Let $\zeta_k \triangleq \sum_{j=D_{k-1}}^{D_k-1}\|e_j\|_2^2$. In the following we show that the random variables $\zeta_k$ are identically distributed and $\zeta_k$ is independent of $\{\zeta_i, i\leq k-2\}$. We note that, for $j = D_{k-1}$ we have $\Delta_{D_{k-1}} = Y_{k-1}$, because at the departure instant of sample $(k-1)$ AoI reduces to sample $(k-1)$'s waiting time in queue, which is zero, plus its transmission time $Y_{k-1}$. Since AoI increases by one in each time slot until next departure it is easy to see that for $j = D_k - 1$, $\Delta_{D_k - 1} = Y_{k-1} + G_k + Y_k - 1$. Therefore, for $D_{k-1} \leq j \leq D_k-1$, we have 
\begin{align*}
e_j = \sum_{i=1}^{\Delta_j} A^{i-1} W_{j-i} = \sum_{i=1}^{Y_{k-1} + j - D_{k-1}} A^{i-1} W_{j-i}.
\end{align*}
From the above analysis and using a change of variable, we obtain
\begin{align}\label{eq11:thm1}
\zeta_k = \sum_{j=D_{k-1}}^{D_k-1}\|e_j\|_2^2 = \sum_{j=Y_{k-1}}^{Y_{k-1} + G_k + Y_k - 1}\|e_{jk}\|_2^2, 
\end{align}
where 
\begin{align*}
e_{jk} = \sum_{i=1}^{j}A^{i-1}W_{D_{k-1} + j - Y_{k-1}  - i}.
\end{align*}
Note that $\zeta_k$ is a function of $Y_{k-1}$, $G_k$, $Y_k$ and $\{W_i,D_{k-1}-Y_{k-1} \leq i \leq D_{k-1}+G_k+Y_k-1\}$. It is easy to see that $\zeta_k$ are identically distributed. Further, since $Y_k$ are i.i.d., $G_k$ is a function of $Y_{k-1}$ and $W_i$ are i.i.d, we infer that $\zeta_k$ is independent of $\{\zeta_i, i\leq k-2\}$. 

Given the above properties of $\{\zeta_k\}$ and $\{D_k - D_{k-1}\}$, we use the extended renewal reward theorem [Theorem 2~\cite{champati:Infocom2019}] and obtain the following result.
\begin{align}\label{eq2:thm1}
\lim_{K \rightarrow \infty} \frac{1}{D_K} \sum_{k=1}^{K}\zeta_k = \frac{\mathbb{E}[\zeta_k]}{\mathbb{E}[D_k - D_{k-1}]}\; \text{a.s.}
\end{align}
Using~\eqref{eq2:thm1}, \eqref{eq11:thm1} and~\eqref{eq1:thm1}, we obtain
\begin{align}\label{eq21:thm1}
\lim_{n \rightarrow \infty} \frac{1}{n} \sum_{j=0}^{n-1}\|e_j\|_2^2 = \frac{\mathbb{E}[\sum_{j=Y_{k-1}}^{Y_{k-1} + G_k + Y_k - 1}\|e_{jk}\|_2^2]}{\mathbb{E}[Y_k + G_k]}\; \text{a.s.}
\end{align}

For $1 \leq r \leq d$, let $e_{jk}(r)$ denote the $r$th element of the vector $e_{jk}$. Let $a^{(i)}(r,l)$ denote the element in $r^{\text{th}}$ row and $l^{\text{th}}$ column of the matrix $A^i$, and $w_i(r)$ denote the $r$th element of the vector $W_i$. Note that, for all $r$, $e_{jk}(r)$ is a linear combination of i.i.d. Gaussian random variables. Therefore $e_{jk}(r)$ is Gaussian with zero mean and variance 
\begin{align}\label{eq3:thm1}
\sigma^2_{e_{jk}(r)} = \sigma^2 \sum_{i=1}^j \sum_{l=1}^{d} [a^{(i-1)}(r,l)]^2.
\end{align}
This implies that for any $j$ and $k$, $\|e_{jk}\|_2^2$ is $\chi^2$-distributed and 
\begin{align}\label{eq4:thm1}
\mathbb{E}[\|e_{jk}\|^2] = \sum_{r=1}^{d} \sigma^2_{e_{jk}(r)}. 
\end{align}
Now, we have 
\begin{align}\label{eq5:thm1}
&\mathbb{E}\left[\sum_{j=Y_{k-1}}^{Y_{k-1} + G_k + Y_k - 1}\|e_{jk}\|_2^2 \right] \nonumber\\
&=\mathbb{E} \left[\mathbb{E}\left[\sum_{j=Y_{k-1}}^{Y_{k-1} + G_k + Y_k - 1}\|e_{jk}\|_2^2 \middle| Y_{k-1},Y_k \right]\right]\nonumber \\
&= \sigma^2 \mathbb{E}\left[\sum_{j=Y_{k-1}}^{Y_{k-1} + G_k + Y_k - 1} \sum_{i=1}^j \sum_{r=1}^{d} \sum_{l=1}^{d}  [a^{(i-1)}(r,l)]^2 \right] \nonumber \\
&= \mathbb{E}\left[\sum_{j=Y_{k-1}}^{Y_{k-1} + G_k + Y_k - 1} \sum_{i=0}^{j-1} \textsf{Tr}\left(A^{i^\top}A^{i}\Sigma\right)\right]
\end{align}
In the second step above, we have used~\eqref{eq3:thm1} and~\eqref{eq4:thm1}. In the last step we have used $\Sigma = \sigma^2 \mathbf{I}_{d}$. From from~\eqref{eq21:thm1} and~\eqref{eq5:thm1}, we have, almost surely,
\begin{align}\label{eq6:thm1}
&\lim_{n \rightarrow \infty} \frac{1}{n} \sum_{j=0}^{n-1}\|e_j\|_2^2 \nonumber\\
 &= \frac{\mathbb{E}\left[\sum_{j=Y_{k-1}}^{Y_{k-1} + G_k + Y_k - 1} \sum_{i=0}^{j-1} \textsf{Tr}\left(A^{i^\top}A^{i}\Sigma\right)\right]}{\mathbb{E}[Y_k + G_k]}.
\end{align}
To finish the proof it is sufficient to show that $\mathbb{E}[f(\Delta)]$ is equal to~\eqref{eq6:thm1}. Since $\Delta_n$ is stationary and ergodic, by Birkhoff’s ergodic theorem~\cite{Kumar_2004}, almost surely,
\begin{align}\label{eq7:thm1}
\mathbb{E}[f(\Delta)] &= \lim_{n \rightarrow \infty} \frac{1}{n} \sum_{j=0}^{n-1} f(\Delta_j) \nonumber\\
&= \lim_{K \rightarrow \infty} \frac{1}{D_K} \sum_{k=1}^{K}\sum_{j=D_{k-1}}^{D_k-1} f(\Delta_j).
\end{align}
To simplify~\eqref{eq7:thm1}, we use steps similar to that used in deriving~\eqref{eq21:thm1} from~\eqref{eq1:thm1}, and obtain
\begin{align*}
\lim_{K \rightarrow \infty} \frac{1}{D_K} \sum_{k=1}^{K}\sum_{j=D_{k-1}}^{D_k-1} f(\Delta_j) = \frac{\mathbb{E}\left[\sum_{j=Y_{k-1}}^{Y_{k-1} + G_k + Y_k - 1} f(j)\right]}{\mathbb{E}[Y_k + G_k]}.
\end{align*}
Hence the result is proven.

\end{document}